\documentclass[12pt,draftcls,onecolumn]{IEEEtran}
\usepackage{epsfig}
\usepackage{graphicx}
\usepackage{cite}

\usepackage{delarray}
\newcommand{\N}{I\!\!N}
\newcommand{\R}{I\!\!R}

\ifCLASSINFOpdf
\else
\fi

\usepackage{graphicx,mathptm}
\newtheorem{definition}{\bf Definition}
\newtheorem{proposition}{\bf Proposition}
\newtheorem{lemma}{\bf Lemma}
\newcommand{\Prob}{I\!\! P}
\newcommand{\E}{I\!\! E}
\begin{document}
%
\title{\LARGE  Evolutionary Poisson Games for Controlling Large Population Behaviors
\vspace{-4mm}}

\author{Yezekael Hayel$^{1,2}$ and Quanyan Zhu$^{2}$\\{~}
{\small $^{1}$LIA/CERI, University of Avignon, Avignon, France, Email: yezekael.hayel@univ-avignon.fr}\\{~} 
{\small $^{2}$Department of Electrical and Computer Engineering, New York University, USA. E-mail: \{yezekael.hayel,quanyan.zhu\}@nyu.edu.}\vspace{-4mm}}


\maketitle

\begin{abstract}
Emerging applications in engineering such as crowd-sourcing and (mis)information propagation involve a large population of heterogeneous users or agents in a complex network who strategically make dynamic decisions. In this work, we establish an evolutionary Poisson game framework to capture the random, dynamic and heterogeneous interactions of agents in a holistic fashion, and design mechanisms to control their behaviors to achieve a system-wide objective. We use the antivirus protection challenge in cyber security to motivate the framework, where each user in the network can choose whether or not to adopt the software. We introduce the notion of evolutionary Poisson stable equilibrium for the game, and show its existence and uniqueness.  Online algorithms are developed using the techniques of stochastic approximation coupled with the population dynamics, and they are shown to converge to the optimal solution of the controller problem. Numerical examples are used to illustrate and corroborate our results.

\end{abstract}


%
\IEEEpeerreviewmaketitle

\section{Introduction}

Emerging engineering applications such as social networks \cite{jackson2008social}, crowdsourcing \cite{kittur2008crowdsourcing,howe2006rise} and the Internet of Things (IoTs) \cite{atzori2010internet} involve a large population of heterogeneous devices or users. These agents interact with each other in a complex environment, in which each agent makes strategic and dynamic decisions in response to the group of agents it interacts with. The group of agents can be random and changing over time. One illustrative example is 5G wireless communication networks \cite{ohmori2000future}. Each mobile can communicate with a number of heterogeneous devices at different times, and makes an investment decision on antivirus software. 
This situation is also analogous to the epidemic spread of influenza in which individual person makes a decision on vaccination. The objective from the perspective of the system designer or government agency is to control the large population behaviors, and induce desirable outcome that is conducive for the sustainable growth of the population. In order to address this issue, the first step of this research is to establish an integrated system framework that allows capturing the random, dynamic and heterogeneous features of the population. 

One useful tool to describe the dynamic evolution of population is  evolutionary game theory \cite{HofBook,Hofbauer03}, which often assumes homogeneous and pairwise interactions between agents. This underlying assumption makes the classical framework insufficient to capture the network properties of the agents, and the heterogeneous local interactions among the players. 
In this paper, we develop an evolutionary Poisson game framework which bridges the gap between the evolutionary game theory with the heterogeneity of the population. We enrich the game-theoretic model by incorporating network topology, the size of the population, and the epidemic process to establish a holistic framework that can be used to address the engineering applications of interest. These unique aspects of the model lead to a customized  evolutionary stability equilibrium concept, and its corresponding replicator dynamics for describing the evolution of the population. 

\begin{figure}[htbp]
\centering
\vspace{-3mm}\includegraphics[width=5in]{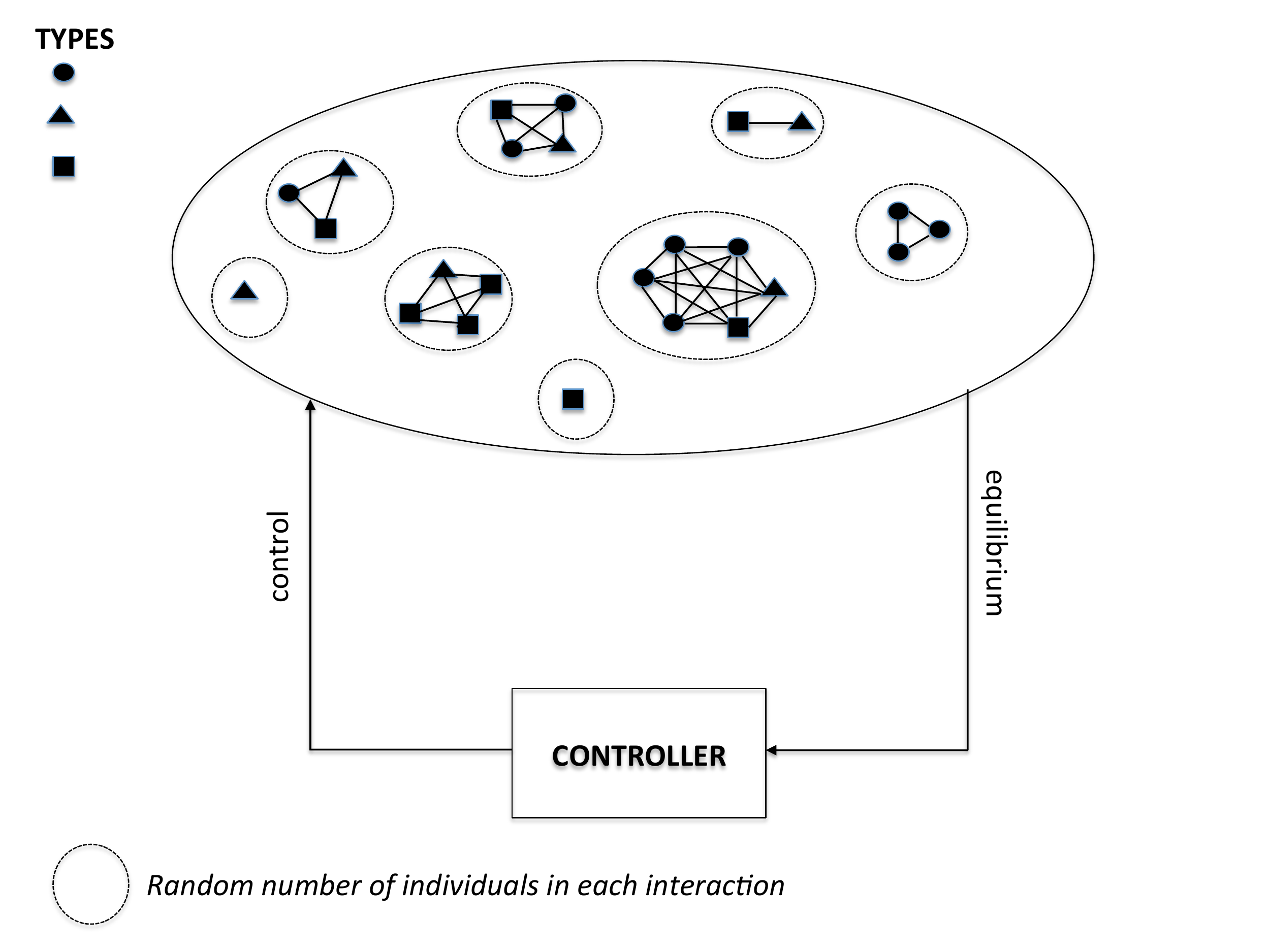}\vspace{-5mm}
\caption{Feedback system between controller and the evolutionary Poisson game model. Each circle represents a group of agents who interact through a network. Each agent is denoted by its type. The controller observes the population states and controls the population to achieve a global objective. }
\label{fig_ex1}\vspace{-2mm}
\end{figure}

The overarching goal of this work is to control the behaviors of the large population to achieve a system-wide objective. Building on the evolutionary Poisson game framework, we leverage the techniques of stochastic approximation to develop an online learning algorithm. Fig. \ref{fig_ex1} illustrates the interdependencies between the game-theoretic model and the controller. The controller observes the population states and inputs a control to drive the system to a global optimum.

The convergence analysis of the controller involves the understanding of the coupled dynamics between the population and the learning algorithm. We show that the convergence is guaranteed under time scale separation and mild conditions on the step sizes. In addition, we use virus protection over a large-scale network as a motivating application to illustrate the link between the game-theoretic model and the application. We fully characterize the global control of the virus protection problem in network systems, and corroborate the results with numerical examples. We observe a phenomenon of heterogeneity induced confidence in which the protection rate decreases as heterogeneity of the population exceeds a certain threshold. 

\subsection{Related Work}
Large population behaviors have been investigated using models from evolutionary games \cite{Hof79,HofBook,Hofbauer03,Tembine08}, Poisson games \cite{Myerson00,Myerson98}, mean-field games \cite{lasry2007mean,MeanField,CDCmean}. These models have successfully captured different aspects of the large population. The evolutionary Poisson game developed in this work integrates the features of evolutionary games and Poisson games to form a more powerful framework to analyze and control systems with a large population of agents. 

From the perspective of applications, the optimal protection problem and epidemics spreading over networks has been recently investigated in \cite{Jiang11,Saha14,Theo13,VM09}. The existing models are not sufficient yet to incorporate the topological information into a holistic epidemic and game model simultaneously. 
 In this work, we aim to address this issue by proposing an integrated framework that can be used to broaden the scope of the applications and capture more pertinent features of the problem. 


\subsection{Organization of the Paper}
The paper is organized as follows: In section \ref{section:model}, we describe the system model, and we develop in section \ref{section:ESS} the evolutionary stable equilibrium concept and its associated replicator dynamics. In Section \ref{section:control}, we present a global control problem, and we provide in section \ref{section:comp} explicit results that completely characterize a class of virus protection problem. Finally, we conclude the paper in section \ref{section:conc} and discuses future work. Due to the page limit of the paper, we remove the proofs of the results. The details of the proofs can be found in \cite{HayelZhu}.


\section{System Model}\label{section:model}
In this section, we introduce 
the large population game model, and discuss an application of epidemic protection in an heterogeneous population. 
\subsection{Random player game}
Large scale interacting systems often involve a random number of interacting players.
 Poisson game  is a natural framework to capture this phenomenon \cite{Myerson98,Myerson00}, which has been successfully used to study decentralized resource allocation in networks \cite{Rodriguez13,simhon14}. A Poisson Game $\Gamma$ is mathematically defined by a five-tuple $(\lambda,\tau,r,\mathcal{C},u)$ where:
\begin{itemize}
\item $\lambda$ corresponds to the mean number of players, typically $\lambda>>1$,
\item $T$ is the set of types of players and each one belongs to one type $t \in T$,
\item The probability of a player being of type $t$ is given by $r(t)$, and the number of players of type $t$ is a Poisson random variable with parameter $\lambda r(t)$,
\item $\mathcal{C}$ is the set of all pure actions available to all players,
\item The utility of a player of type $t$ is $u_t(a,x)$ with $a$ is the pure action, and $x$ is a vector of size $|\mathcal{C}|$ where $x(b)$ is the number of players who choose action $b$ in $\mathcal{C}$. 
\end{itemize}  
The expected utility of a player of type $t$ who plays action $a$ while the rest of the players are expected to play using strategy $\sigma$ is:
$$
U_t(a,\sigma)=\sum_{x \in Z(\mathcal{C})}P(x|\sigma)u_t(a,x),
$$
where $Z(\mathcal{C})$ denotes the set of elements $w \in \R^{\mathcal{C}}$ such that $w(c)$ is a non-negative integer for all $c \in \mathcal{C}$.
The decomposition property of the Poisson distribution yield:
$$
P(x|\sigma)=\prod_{b\in \mathcal{C}}e^{-\lambda \tau(b)}\frac{(\lambda \tau(b))^{x(b)}}{x(b)!},
$$
and
$$
\tau(b)=\sum_{t \in T}r(t)\sigma_t(b).
$$
If players play according to the strategy $\sigma$, $\sigma_t(c)$ is the probability that a player of type $t$ chooses the pure action $c$. Finally, the expected utility whether a player chooses action $\theta \in \Delta(\mathcal{C})$ is:
$$
U_{t}(\theta,\sigma)=\sum_{a \in \mathcal{C}}\theta(a)U_t(a,\sigma).
$$

\begin{definition}
The strategy $\sigma^*$ is a pure Nash equilibrium if
$$
\forall t \in T, \quad \sigma^* \in B_t(\sigma^*),
$$
with
$$
B_t(\sigma)=\{b \in \mathcal{C}: b \in \arg\min_{a \in \mathcal{C}}U_t(a,\sigma) \}.
$$
\end{definition}
We can also extend this framework to the mixed-strategy Nash equilibrium by considering the set of mixed best responses $\Delta(B_t(\sigma))$.

\subsection{Application to epidemic protection in heterogeneous population}
Having defined the non-cooperative game in the context of heterogeneous  interacting randomly individuals, we describe one application related to virus protection. Many recent work works have ignored 
the topology of the interaction, the heterogeneity of the individuals or the selfishness of their decision in their models. In this work, we develop a holistic framework that can incorporate these features. We start by introducing the Susceptible-Infected-Susceptible (SIS) epidemic model, which has been well-studied in the literature \cite{BaileySIS} and recently has gained lots of interests for modeling computer viruses propagation \cite{Prakash10, Kephart91}.

Consider an SIS epidemics over a graph, there exists a limiting spreading factor rate, denoted by the critical epidemic threshold, below which the infection vanishes exponentially fast in time, and above which the critical threshold the network stays infected. In an in-homogeneous SIS epidemics, we can express the epidemic threshold of an individual effective spreading rate $\tau_i$ of each node $i$. Indeed, it is has been shown in \cite{PVM14} that by using a mean-field approximation of the Markov process model for the epidemic, for the complete graph structure with $N$ nodes, the critical threshold  thus satisfies the following relation:
$$
\sum_{i=1}^{N}\frac{1}{1+\tau_i^c}=N-1.
$$
The contamination process of our SIS epidemic is a Poisson variable with rate $\beta$ but our spreading framework is in-homogeneous as we consider that each node of type $t$ has a recovery process with rate $\delta_t$. 


Our framework is enough generic such that its can be applied to the control of large complex systems, such as  virus spreading \cite{Prakash10} and information cascading \cite{watts2002simple}. In order to illustrate our framework, we describe in the next section, our model for the controlled of an in-homogeneous SIS over a large population in which the interaction structure is stochastic.

We consider an individual protection strategic game where each player has a type, or private information which determines his recovery capability (e.g. the rate of recovery), and incomplete information (e.g. nodes are not aware of the number of players they interact with \footnote{Each local interaction has a random number of players which follows a Poisson process with an average population equal to $\lambda$($\lambda>>1$), i.e., most of the interactions that occur in the population  involve a large number of interacting users.}).

The set of pure actions of the players is $\mathcal{C}=\{OFF,ON\}$, and the set $Z(\mathcal{C})=\N^2$. Players are characterized by their recovery rate $\delta_t$ which depends on their type $t$. Then, the effective spreading rate for each type $t$ player is $\tau_t=\frac{\beta}{\delta_t}$. Based on the expression of the critical epidemic threshold for the heterogeneous SIS in a complete graph, we obtain a following necessary and sufficient condition over the effective spreading rates $\tau_t$ and the number of nodes $x_t$ of type $t$ that do not invest, in order for the infection to propagate in a complete graph:
\begin{eqnarray}
\sum_{t=1}^{T}\frac{x_t}{1+\tau_t}\leq \sum_{t}^{T}x_t-1.
\end{eqnarray}
This inequality is equivalent to the linear constraint:
\begin{eqnarray}
\sum_{t=1}^{T}\frac{x_t\tau_t}{1+\tau_t}\geq 1.
\end{eqnarray}
Depending on the decision of each player, if the infection is propagated over the entire network, then the cost for a player that does not protect itself is $K$, otherwise its cost is $0$. This cost may represent the restoring cost when a node is contaminated, or also this cost can be a penalty proposed by the system designer in order to control self-protection behavior. Then, the utility of a player of type $t$ is given by:
$$
u_t(OFF,(x_1,\ldots,x_T))=
 \left\{ \begin{array}{l l}
K& \mbox{if} \quad \sum_{t=1}^{T}\frac{x_t\tau_t}{1+\tau_t}\geq 1, \\
0 & \mbox{otherwise.}
\end{array}\right.
$$ 
 In our framework, the type-$t$ utility function is not defined over the set $Z(\mathcal{C})$ as in \cite{Myerson98} but as follows:
$$
u_t:\mathcal{C} \times \underbrace{Z(\mathcal{C}) \times \ldots \times Z(\mathcal{C}) }_{\times T} \rightarrow \R.
$$
In fact, the utility function in a Poisson game should depend on the total number of players choosing the same action over different types. In the model, we do not have such aggregative assumption in the utility, and the type as an impact on the utility function. For the same number of individuals that do not invest, i.e., take the action $OFF$, the utility of a player depends on the number of such individuals of each type. We then use the concept of random player game model proposed in \cite{Milchtaich04}, which is a generalization of Poisson games. In fact, games with a random number of players are a natural extension of Bayesian games, which are a class of incomplete information games \cite{Harsanyi}.

We denote by $X_t$ the random variable which determines the number of players of type $t$ that do not invest. Based on the decomposition property of the Poisson distribution, $X_t$ is a Poisson distribution with parameter $\lambda r(t)\sigma_t(OFF)$. Then, the total number of players that do not invest is a Poisson distribution with parameter $\lambda \sum_t r(t)\sigma_t(OFF)$.
If  a node decides to be protected, he pays a cost $C$, i.e.,
$$
u_t(ON,(x_1,\ldots,x_T))=C.
$$ 

Note that the utility functions do not depend on the type $t$ of the user \footnote{A similar analysis can be done for the case of type-dependent equilibrium (\cite{Fudenberg}).}. We consider a symmetric Nash Equilibrium. Denote by $p$ the probability that a player (of any type) chooses action $OFF$, i.e., for all types $t=1,\ldots,T$, $\sigma_t(OFF)=p$ and $\sigma_t(ON)=1-p$. The expected utility of a player who plays pure action $OFF$ while all other players are expected to play according to a mixed strategy $p$ depends on the realization vector $\textbf{x}=(x_1,x_2,\ldots,x_T)$ of the random vector $\textbf{X}=(X_1,\ldots,X_T)$ by:
$$
U_t(OFF,p)=\sum_{\textbf{x} \in \N^T}P(\textbf{X}=\textbf{x}|p)u_t(OFF,\textbf{x}):=U(OFF,p).
$$
Based on the decomposition and aggregation properties of the Poisson distribution, we arrive at
$$
P(\textbf{X}=\textbf{x}|p)=\prod_{t=1}^{T}P(X_t=x_t|p)=\prod_{t=1}^{T}\frac{(\lambda r(t) p)^{x_t}}{x_t!}e^{-\lambda r(t) p}.
$$
Then, the expected utility of a player that does not invest in protection, in face of a population profile $p$ is given by:
\small
\begin{eqnarray*}
U(OFF,p)&=&\sum_{\textbf{x} \in \N^T}P(\textbf{X}=\textbf{x}|p)u(OFF,\textbf{x})\\
&=&K \left(\sum_{\textbf{x}:\sum_{t=1}^{T}\frac{x_t\tau_t}{1+\tau_t}\geq 1}P(\textbf{X}=\textbf{x}|p)\right)\\
&=&K \left(1-\sum_{\textbf{x}:\sum_{t=1}^{T}\frac{x_t\tau_t}{1+\tau_t}< 1}P(\textbf{X}=\textbf{x}|p)\right)\\
&=&K \left(1-\sum_{\textbf{x}:\sum_{t=1}^{T}\frac{x_t\tau_t}{1+\tau_t}< 1}\prod_{t=1}^{T}P(X_t=x_t|p)\right)\\
&=&K \left(1-e^{-\lambda p}\sum_{\textbf{x}:\sum_{t=1}^{T}\frac{x_t\tau_t}{1+\tau_t}< 1}(\lambda p)^{\sum_{t=1}^{T}x_t}\prod_{t=1}^{T}\frac{r(t)^{x_t}}{x_t!}\right).
\end{eqnarray*}
\normalsize
Finally, the expected utility from playing action $q \in \Delta (\mathcal{C})$ is given by:
$$
U(q,p)=qU(OFF,p)+(1-q)C.
$$
Based on Definition 1, a (symmetric) mixed Nash Equilibrium $p^*$ for the protection game with a random number of player satisfies:
$$
\forall q \neq p^*, \quad U(p^*,p^*) \leq U(q, p^*).
$$
\begin{lemma}\label{lemma1}
If $C \geq K$, then the pure Nash equilibrium is $p^*=1$. 
\end{lemma}

\begin{proof}
The proof of this lemma follows the intuition that if the cost for being protected $C$ is higher than the cost of being infected $K$, then every user will take the risk to be infected. Indeed, if it is, the cost injured is less or equal to the cost if it was protected. Mathematically speaking, for any population profile $p$ the utility of an individual to be protected is $C$ and we have that:
$$
U(OFF,p)=K \left(\sum_{\textbf{x}:\sum_{t=1}^{T}\frac{x_t\tau_t}{1+\tau_t}\geq 1}P(\textbf{X}=\textbf{x}|p)\right)<K.
$$
Then if we have $C \geq K$, then:
$$
\forall p, \quad U(OFF,p)<U(ON,p).
$$
Thus the strategy $OFF$ is a dominant strategy and all individuals play this action at equilibrium.
\end{proof}

We next state the proposition that describes the mixed equilibrium. 

\begin{proposition}[Existence and Uniqueness]\label{exun}
If the parameters of the system $\lambda$, $C$, $K$ (with $C<K$), $T$, the type distribution $r(\cdot)$ and the effective spreading rates $\tau_t$ satisfy the following condition:
$$
\sum_{\textbf{x}}\prod_{t=1}^{T}\frac{(\lambda r(t))^{x_t}}{x_t!}<(1-\frac{C}{K})e^{\lambda}.
$$
Then, there exists one unique mixed Nash equilibrium. 
\end{proposition}

\begin{proof}
A mixed Nash equilibrium $p^*$ satisfies:
$$
\forall q \neq p^*, \quad U(p^*,p^*) \leq U(q, p^*),
$$
which is equivalent to:
$$
p^* \in \arg\max_{p} U(p,p^*)=\arg\max_{p} (pU(OFF,p^*)+(1-p)C).
$$
Then, we study the solution $\tilde{p}$ of the equation:
$$
K \left(1-e^{-\lambda p}\sum_{\textbf{x}:\sum_{t=1}^{T}\frac{x_t\tau_t}{1+\tau_t}< 1}(\lambda p)^{\sum_{t=1}^{T}x_t}\prod_{t=1}^{T}\frac{r(t)^{x_t}}{x_t!}\right)=C,
$$
which is equivalent to:
\begin{eqnarray}
F(p):=\sum_{\textbf{x}:\sum_{t=1}^{T}\frac{x_t\tau_t}{1+\tau_t}< 1}(\lambda p)^{\sum_{t=1}^{T}x_t}\prod_{t=1}^{T}\frac{r(t)^{x_t}}{x_t!}=(1-\frac{C}{K})e^{\lambda p}:=G(p).
\label{sol}
\end{eqnarray}
Note that both functions are continuous, strictly increasing over the interval $[0,1]$ and also we have:
$$
G(0)=(1-\frac{C}{K})<1=F(0).
$$
The function $G(.)$ is strictly convex over the interval $[0,1]$. Also the same for the function  $F(.)$ which is a finite sum of strictly convex functions and then it is a strictly convex function over the interval $[0,1]$. Then, if we assume that $F(1)<G(1)$ there exists almost one solution of equation (\ref{sol}) inside the interval $]0,1[$. More as the left-hand side function is a polynom and the right-hand side an exponential, this equation, if it has a solution over $[0,1]$, this solution is unique.
\end{proof}

We have proved the existence and uniqueness of the equilibrium depending on the parameters of the problem. Based on simple geometric argument, as function $F$ and $G$ are continuous over the interval $[0,1]$, we can show that the equilibrium $p^*$ is strictly decreasing in $C$ (i.e., more expensive is the protection, less individuals will adapt the strategy $OFF$.), and it also is strictly decreasing with $\lambda$ (i.e., the average number of individuals in each local interaction). Indeed, having more people in average at each interaction makes individuals  more vulnerable, and then the protection rate (i.e., proportion of individuals protected) becomes higher at the equilibrium.
A last important remark from the previous proposition is that the strategy $ON$ cannot be a dominant strategy for any values of the parameters of the model. In fact, $F(0)>G(0)$, and there is always a proportion of individuals that are not protected.

\section{Evolutionary stability and dynamics}\label{section:ESS}

Another equilibrium concept, which is more robust than the Nash equilibrium and well adapted to large population games is called {\it Evolutionary Stable Strategies} (ESS). It is based on evolutionary principles that have been originally defined in \cite{MaynSm-Price} in biology, and recently applied to engineering and systems \cite{Tembine08}. In this section, we will introduce the concept of ESS and its associated dynamics.

\subsection{Evolutionary stability concept} 

ESS is a strategy such that, if adopted by all the players, is robust against deviations of a (possibly small) fraction of the population. From a biological point of view, it can be seen as a generalization of Darwin's idea of survival of the fittest, while from a game theory perspective, it is a refinement of the Nash Equilibrium, which satisfies a stability property. 

It is important to note that a mixed strategy $p$ can be also interpreted as the set of distributions of pure strategies among the players \cite{HofBook}. In our setting, the mixed strategy $p$ can describe, assuming that each player plays a pure action in $\mathcal{C}$, the proportion of players that choose the action $OFF$. This is also called the strategy profile of the population. Having this equivalent point of view of the game in mind, an ESS, if adopted by the whole population, is \textit{resistant} against mutations of a small fraction of individuals in the population. Suppose that the whole population adopts a strategy $q$, and that a fraction $\epsilon$ of \textit{mutants} deviate to strategy $p$. Strategy $q$ is an ESS if $\forall p\neq q$, there exists some $\epsilon_p>0$ such that $\forall \epsilon\in (0,\epsilon_p)$:
\begin{equation}
 U(q,\epsilon p+(1-\epsilon)q)<U(p,\epsilon p+(1-\epsilon)q).
 \label{equadef}
\end{equation}
In other words, this strict inequality says that an ESS defeats any small mutations (relative to $\epsilon$) of the population profile. In that sense, the equilibrium concept of ESS is said to be more robust than the Nash Equilibrium, because it is robust against the deviation of a fraction of players, and not only one. 

Another approach to study the evolutionary stability of an equilibrium is to consider the dynamics of the strategies inside the population. When the utility function is bilinear, this strict inequality condition can be replaced by two simple conditions: Nash condition and stability conditions. This type of analysis which makes the game as a standard evolutionary game that has been proposed in \cite{Tembine08} to analyze the evolutionary stability with a random number of individuals at each interaction. But in our setting, the utility function is clearly not bilinear and then we cannot use the Nash and stability conditions instead of equation (\ref{equadef}).

\subsection{Dynamics}

Another way to describe how a population reaches a stable situation is through the replicator dynamics, which serves to highlight the role of selection from a dynamic perspective. It is formalized by a system of ordinary differential equations, and it establishes that the evolution of the size of the populations depends on the fitness they achieve during interactions. A strategy will sustain if its fitness is higher than the fitness averaged over all the strategies used in the whole population. The \textit{folk theorem of evolutionary games} allows to establish a strict connection between the stable points of the Replicator Dynamics and the Nash Equilibria \cite{Hof79}.

In our context, the RD equation can be formalized as follows:
\begin{eqnarray}
\dot{p}(t)&=&p(t)\left[U(p(t),p(t))-U(OFF,p(t))\right],\nonumber\\
&=&p(t)(1-p(t))[U(ON,p(t))-U(OFF,p(t))],\nonumber\\
&=&p(t)(1-p(t))[C-U(OFF,p(t))].
\label{repdyn}
\end{eqnarray}

We have the results from the folk theorem of the replicator dynamics in population games that an interior rest point of the dynamics is a Nash equilibrium of our game. Then, we have another method to describe the equilibrium, as a rest point of a dynamical system. This provides a very important insight for a global control of the system. We will see in the next section, that this approach gives us the possibility to determine a two time-scale process  to optimize a global performance of the system without computing explicitly the Nash equilibrium. Moreover, we prove in the next proposition that any rest point of our dynamics is not only a Nash equilibrium but also an ESS. In our setting, we have the following result.
\begin{proposition}
If our game has a unique mixed equilibrium $p^*$, then it is an ESS and it corresponds to the unique interior rest point of the ODE (\ref{repdyn}).
\end{proposition}

\begin{proof}
Based on proposition (\ref{exun}), we first assume that the parameters of the system are such that there exists a mixed nash equiliiubrm $p^*$ and we have proved that then it is unique. More, this equilibrium is the unique exterior rest point of the replicator dynamics:
$$
\dot{p}(t)=p(t)(1-p(t))[C-U(OFF,p(t))].
$$
We have that $U(OFF,p^*)=C$. In order to prove that $p^*$ is an ESS, we choose another mixed strategy $q$ adopted by a proportion $\epsilon$ of individuals. Then, based on equation (\ref{equadef}), $p^*$ is an ESS if the exists an $\epsilon_q$ such that $\forall \epsilon \in (0,\epsilon_q)$, we have:
$$
U(p^*,\epsilon q+(1-\epsilon)p^*)<U(q,\epsilon q+(1-\epsilon)p^*).
$$ 
In fact, we prove this inequality for $\epsilon_q=1$. Let first assume w.l.o.g. that $0<q<p^*$. A similar analysis can be done for the case when $1>q>p^*$. We fixed  $\epsilon \in (0,1)$ and we denote $p_{\epsilon}=\epsilon q+(1-\epsilon)p^*$. Note that $q\leq p_{\epsilon}<p^*$ (if $\epsilon=0$ there is no mutant, so it is not an interesting case). After some simple algebraic decompositions, we have the following equivalence:
$$
U(p^*,\epsilon q+(1-\epsilon)p^*)<U(q,\epsilon q+(1-\epsilon)p^*),
$$
rewriting as:
$$
U(p^*,p_{\epsilon})<U(q,p_{\epsilon}),
$$
is equivalent to
$$
(p^*-q)(U(OFF,p_{\epsilon})-C)<0.
$$
As the mixed equilibrium $p^*$ is unique and is the unique interior rest point of the replicator dynamics, for all $p<p^*$ (resp. $p>p^*$) we have that $\dot{p}(t)>0$ (resp. $\dot{p}(t)<0$). We have that $p_{\epsilon}<p^*$ and then $\dot{p_{\epsilon}}(t)>0$ which, based on the replicator dynamics o.d.e. (\ref{repdyn}), means that for any time $t$:
$$
C-U(OFF,p_{\epsilon}(t))>0 \Longleftrightarrow C>U(OFF,p_{\epsilon}(t)).
$$
Finally, as it is try for any time $t$, it is also true for any value $p_{\epsilon}$ such that  $q\leq p_{\epsilon}<p^*$ and then:
$$
(p^*-q)(U(OFF,p_{\epsilon})-C)<0,
$$
which proves that $p^*$ is an ESS.
\end{proof}

The discrete-time version of the replicator dynamics is described as follows:
\begin{eqnarray}
p_{n+1}=p_{n}+b(n)p_n(1-p_n)[C-U(OFF,p_n)].
\label{repdynd}
\end{eqnarray}
Let timescale parameter $b(n)$ be
$$
\sum_{n}b(n)=\infty \quad\mbox{and}\quad \sum_{n}b(n)^2<\infty,
$$
then, as $n$ tends to infinity, the discrete-time iteration (\ref{repdynd}) is an approximation of the replicator dynamics ODE given by  (\ref{repdyn}).

\section{Online global control}\label{section:control}

The overarching goal of this work is to design a global control for the entire heterogeneous, stochastic interacting population. As an example of global objective function for the controller (see Fig. \ref{fig_ex1}), we consider in this section the global revenue. Moreover, we propose an online learning process as we assume that the infected cost $K$ is not known by the controller. This parameter is highly related to how individuals evaluate the cost of being infected, or by definition, it is difficult to estimate it correctly. We then propose an online learning algorithm so that a  controller can optimize a global function of the equilibrium population without knowing this parameter. We use the average revenue as a global function given by
$$
R(C)=\lambda (1-p^*(C))C,
$$  
where $p^*(C)$ is the equilibrium considering the control parameter $C$. Its goal is to maximize this function depending on the price $C$. Based on simulations, we observe that the equilibrium $p^*$ depends on the cost $C$ as an $S$-shaped function, concave and strictly convex. Particularly, for $C$ high enough, i.e., $C>\frac{K}{2}$, the equilibrium is strictly convex in $C$. Then, there exists a price $C_0$ such that for all $C>C_0$, the function $p^*(C)$ is strictly convex. Then, we have the following lemma.

\begin{lemma}\label{lemmaconv}
The revenue $R(C)$ of the provider is strictly concave for $C \in [C_0,K]$.
\end{lemma}

\begin{proof}
We have that $R(C)=\lambda (1-p^*(C))C$. We take the two time derivatives and then:
$$
R'(C)=-\lambda p^*(C)-\lambda C \frac{dp^*}{dC}(C),
$$
and
$$
R''(C)=-2\lambda \frac{dp^*}{dC}(C)-\lambda C \frac{d^2p^*}{dC^2}(C)<0,
$$
because the equilibrium is strictly increasing with $C$ and it is strictly convex over the interval $[C_0,C]$. Then the revenue of the provider is strictly concave over the interval $[C_0,C]$. 
\end{proof}

Based on this result, a gradient algorithm can be used to converge to then optimal price $C^*$ that maximized the revenue of the provider. In fact, we have existence and uniqueness of an optimal price. Designing a gradient algorithm in our context needs to estimate the gradient function of the revenue because there is no closed form expression of the equilibrium function. We then propose to use the following Simultaneous Perturbation Stochastic Approximation (SPSA) \cite{Spall92}:
\begin{eqnarray}
{\frac{d {R}}{dC}}(C)\sim \left(\frac{R(C+\delta \Delta)-R(C)}{\delta \Delta}\right),
\label{approx}
\end{eqnarray}
where $\Delta$ is a random variable such that $\Prob(\Delta=1)=\Prob(\Delta=-1)=\frac{1}{2}$, and $\delta$ is a small constant. One sample form of this approximation is proposed in \cite{Spall98}. However, it has been observed that this single form introduces significant bias, so that in general the two-sample form, given by equation (\ref{approx}), is considered.

Since the equilibrium of the population game is characterized by the rest-point of  the replicator dynamic given by equation (\ref{repdyn}), the controller can set a new price $(C+\delta \Delta)$ to optimize his revenue, and wait for convergence of the replicator dynamics to the equilibrium. Therefore, by observing the equilibrium, the controller has an estimation of its new revenue and also of the gradient of his revenue.


Given that for a fixed price $C$ the replicator dynamics has a global asymptotically stable equilibrium $p^*(C)$, our algorithm converges to the optimal price. Considering a starting price $C_0$, we then consider the following approximate gradient descent algorithm:
\begin{eqnarray*}
\forall n=0,1,\ldots,\quad C_{n+1}&=&C_n+a(n)\left(\frac{\tilde{R}(C_n+\delta \Delta_n)-\tilde{R}(C_n)}{\delta \Delta_n}\right),\\
&=&C_n+a(n)f(C_n,\Delta_n),
\label{stocgrad}
\end{eqnarray*}
where $a(n)$ is the update step size, a discrete random variable $\Delta_n$ at each step $n$ such that $\Prob(\Delta_n=1)=\Prob(\Delta_n=-1)=\frac{1}{2}$, and
$$
\tilde{R}(C_n)=\lambda (1-\tilde{p}(C_n))C_n.
$$
We can then define for all $n$ the following functions:
$$
h(C_n):=\E [f(C_n,\Delta_1)],
$$ 
and 
$$
M_{n+1}:=f(C_n,\Delta_{n+1})-h(C_n).
$$
Then, the approximate stochastic gradient descent can be written as:
$$
C_{n+1}=C_n+a(n)[h(C_n)+M_{n+1}],
$$
where $M_n$ is a martingale difference sequence. The update step-size parameter $a(n)$ has to be efficiently designed in order to guarantee the convergence to the global optimal solution. Indeed, the control update has to wait for the convergence of the replicator dynamics to the equilibrium. Therefore, there is a relation between the speeds of the two dynamical processes.
For a fixed $C$, the equilibrium $p^*(C)$ is a global attractor, which is also a global asymptotically stable equilibrium of the replicator dynamics. Indeed, the SPSA is coupled with the replicator dynamics.
Let $\epsilon$ be a positive constant, and the replicator equation is rewritten as\begin{eqnarray}
\dot{p}(t)&=&\frac{1}{\epsilon}p(t)(1-p(t))[C-U(OFF,p(t))].
\label{repdyneps}
\end{eqnarray}
By having $\epsilon$ small enough, the SPSA views the replicator dynamics as quasi-equilibrated while the replicator dynamics views the SPSA as quasi-static. Then, we can prove the convergence of our SPSA by viewing the underlying replicator dynamics as a two timescale dynamical system. For a sufficient relative speed of convergence of the replicator dynamics compared to the stochastic gradient descent, we can prove that our algorithm converges to the price that optimizes the global objective function. The step size of strategy update  of the population is  $\frac{1}{n}$ in discrete time. Then, we have the following proposition that yields the conditions on the step-size update of the gradient algorithm for reaching the optimal solution.

\begin{proposition}\label{condconv}
If we have the following conditions:
$$
\sum_{n}a(n)=\infty,\quad \sum_{n}a(n)^2<\infty\quad \mbox{and}\quad \frac{1}{na(n)}\rightarrow 0,
$$
then 
$
C_n \rightarrow C^*, \ \mbox{a.s.}
$
\end{proposition}
\begin{proof}
On one side, the approximate gradient descent algorithm follows:
\begin{eqnarray*}
\forall n=0,1,2,\ldots,\quad C_{n+1}=C_n+a(n)(\frac{\tilde{R}(C_n+\delta \Delta_n)-\tilde{R}(C_n)}{\delta \Delta_n}),
\end{eqnarray*}
where the reward depends on the equilibrium as $\tilde{R}(C_n)=\lambda (1-\tilde{p}(C_n))C_n$ and $a(n)$ is the step size of the updating rule. On the other side, the replicator dynamics given by equation $(\ref{repdyn})$ is the limit of a discrete time iteration with an update stepsize $1/n$. Then, the two discrete time iterations are coupled like in \cite{borkar96}. Then, in order to have the convergence of the coupled iteration processes, we need the following conditions:
\begin{enumerate}
\item $sup_{n}(||C_n||)<\infty$, a.s.,
\item for fixed $C$, the replicator dynamics o.d.e has a globally asymptotically stable equilibrium $p^*(C)$,
\item the o.d.e limit of the approximate gradient algorithm has a globally asymptotically stable equilibrium when replacing the equilibrium by $p^*(C)$,
\item $$
\sum_{n}a(n)=\infty,\quad \sum_{n}a(n)^2<\infty\quad \mbox{and}\quad \frac{1}{na(n)}\rightarrow 0,
$$
\end{enumerate} 
The first point is verified as if $C>K$ then $p^*=1$, thus we have that for all $n$, $||C_n||<K$ and then $||R(C_n)||<\lambda K$. Thus we have:
\begin{eqnarray*}
||C_{n+1}||&=&||C_n+\frac{1}{n}(\frac{\tilde{R}(C_n+\delta \Delta_n)-\tilde{R}(C_n)}{\delta \Delta_n})||,\\
&<&||C_n||+||\frac{1}{n}(\frac{\tilde{R}(C_n+\delta \Delta_n)-\tilde{R}(C_n)}{\delta \Delta_n})||,\\
&<&||C_n||+\frac{||\tilde{R}(C_n+\delta \Delta_n)||+||\tilde{R}(C_n)||}{\delta},\\
&<&K(1+\frac{2\lambda}{\delta})<\infty
\end{eqnarray*}
and then $sup_{n}(||C_n||)<\infty$.
The second point is verified as the replicator dynamics rest point, for a fixed $C$, is an interior ESS in our case and then a global globally asymptotically stable equilibrium (see \cite{Hofbauer03}-Theorem 7). Based on the strict concavity of the revenue function proved in lemma \ref{lemmaconv}, this function has a unique maximizer, and then the o.d.e limit of the stochastic gradient has a global globally asymptotically stable equilibrium. This proves the third point. Finally, we assume that   
$$
\sum_{n}a(n)=\infty,\quad \sum_{n}a(n)^2<\infty\quad \mbox{and}\quad \frac{1}{na(n)}\rightarrow 0,
$$
which implies that the gradient update moves on the slower timescale  than the replicator dynamics. Thus we have the convergence almost surely of our approximate gradient descent algorithm given by equation (\ref{stocgrad}) to the optimal price, i.e. $C_n \rightarrow C^*\quad \mbox{a.s.}$.
\end{proof}
The last proposition provides conditions on the time scale of the gradient algorithm to ensure that the coupled algorithm converges to the optimal solution of the global control problem. 

\section{Complete characterization of the virus protection global control problem}\label{section:comp}

In this section, we first provide a complete characterization of the solutions to the global control problem for the population of one single type. Second, we  show the impact of heterogeneity and randomness of our system on the equilibrium results and the global control problem.

\subsection{Population with Single Type}
Consider a fully connected network of size $N$, and all the players are of the same type.  Hence $\tau_i=\tau$ and the value of the epidemic threshold $\tau_c$ is exactly equal to inverse of the largest eigenvalue, i.e. the spectral radius, $\lambda_1$ of the adjacency matrix \cite{VM09}. In particular, we arrive at
$$
\tau_c(x)=\frac{1}{x-1},
$$
where $x>1$ is the number of nodes that do not invest. The special cases where $x=1$  and $x=0$ are not representative of our problem since there is no propagation effects in these two cases. The utility function for a node that does not invest can be simplified into
$$
u_t(OFF,x)=
 \left\{ \begin{array}{l l}
K& \mbox{if} \quad x \geq 1+\frac{1}{\tau}, \\
0 & \mbox{otherwise.}
\end{array}\right.
$$

The value $x$ is not known by every node, but it follows a Poisson distribution with rate $\lambda$. We denote by $\tilde{p}$ the solution to the following equation
$$
\sum_{k=0}^{\lfloor \frac{\delta}{\beta} \rfloor}\frac{(\lambda \tilde{p})^k}{k!}=\left(1-\frac{C}{K}\right)e^{\lambda \tilde{p}}.
$$
The strategy $OFF$ is dominant if $K\leq C$ as stated in Lemma (\ref{lemma1}). Next, we investigate the case where $K>C$. Note that if the effective spreading rate is sufficiently large, i.e., $\tau=\frac{\beta}{\delta}>1$, then, we can find the equilibrium  given by:
$$
p^*=\frac{1}{\lambda}\log \left(\frac{K}{K-C}\right).
$$ 
This solution is not $0$ since there is a very small probability that an individual is not infected. However, this value becomes close to $0$ as $\lambda$ increases.
\begin{proposition}
If the effective spreading rate is high enough but not too much, i.e. $1/2 < \tau \leq 1$, then we have the unique equilibrium given by:
$$
p^*=-\frac{1}{\lambda}\left[1+W\left(-e^{-1}\left(1-\frac{C}{K}\right)\right)\right],
$$
with $W(x)$ is known as the \textit{Lambert W function} and defined such that $x=W(x)e^{W(x)}$.
\end{proposition}
\begin{proof}
Assuming that $1/2 < \tau \leq 1$ implies that $\lfloor \frac{\delta}{\beta} \rfloor =1$ and then the mixed equilibrium is solution bet ween the interval $[0,1]$ of the following equation:
$$
1+\lambda p=(1-\frac{C}{K})e^{\lambda p}.
$$
This equation is equivalent to:
$$
e^{\lambda p}=\frac{\lambda}{1-\frac{C}{K}}(p+\frac{1}{\lambda}).
$$
This equation is the following \textit{transcendental algebraic} equation:
$$
e^{-cp}=a_0(p-r), 
$$
with the following constants:
$$
c=-\lambda, \quad a_0=\frac{\lambda}{1-\frac{C}{K}},\quad \mbox{and}\quad r=-\frac{1}{\lambda}.
$$
Then the solution is given by:
$$
p^*=r+\frac{1}{c}LambertW(\frac{ce^{-cr}}{a_0}),
$$
which gives:
\begin{eqnarray*}
p^*&=&-\frac{1}{\lambda}-\frac{1}{\lambda}LambertW(\frac{-\lambda e^{-1}(1-\frac{C}{K})}{\lambda}),\\
&=&-\frac{1}{\lambda}[1+LambertW(-e^{-1}(1-\frac{C}{K}))].
\end{eqnarray*}
\end{proof}
In the case where $\tau\leq 1/2$, we have the following description of the equilibrium $p^*$:
\begin{equation}
p^*=\left\{\begin{array}{ll} 
\tilde{p} & \mbox{ if } (1-\frac{C}{K})e^{\lambda } \leq  \sum_{k=0}^{\lfloor \frac{\delta}{\beta} \rfloor}\frac{(\lambda)^k}{k!},  \\\nonumber
1 &  \mbox{ otherwise. } 
\end{array}\right.  \nonumber
\label{condmix}
\end{equation}

The solution $\tilde{p}$ is in fact the solution of the so-called \textit{Generalized Lambert W function} \cite{Scott06}:
$$
e^{-cx}=\frac{P_N(x)}{Q_{M}(x)},
$$
where $c>0$ is a constant and $P_N(x)$ and $Q_M(x)$ are polynomials in $x$ of respectively orders $N$ and $M$. Though this equation cannot be solved in its general form (approximations are possible for simple cases \cite{Scott13}), it embodies an interesting link between gravity theory and quantum mechanics \cite{Farrugia07}. 

\subsection{Population with Multiple Types}
We consider  2 types of individuals in a population, i.e., $T=2$. We then have $r(1):=r$ (i.e., the proportion of type-1 individuals in the population) and $r(2):=1-r$ (i.e., the proportion of type-2 individuals). Individuals from each type differ in their capacity (e.g. recovery rate $\tau$) to recover from the virus. We let $\tau_1<\tau_2$, which means that $\delta_1>\delta_2$, i.e., type-1 individuals are more resilient to the virus and it takes generally less time for them to react and then to recover. 

\begin{figure}[htbp]
\centering
\vspace{-3mm}\includegraphics[width=5in]{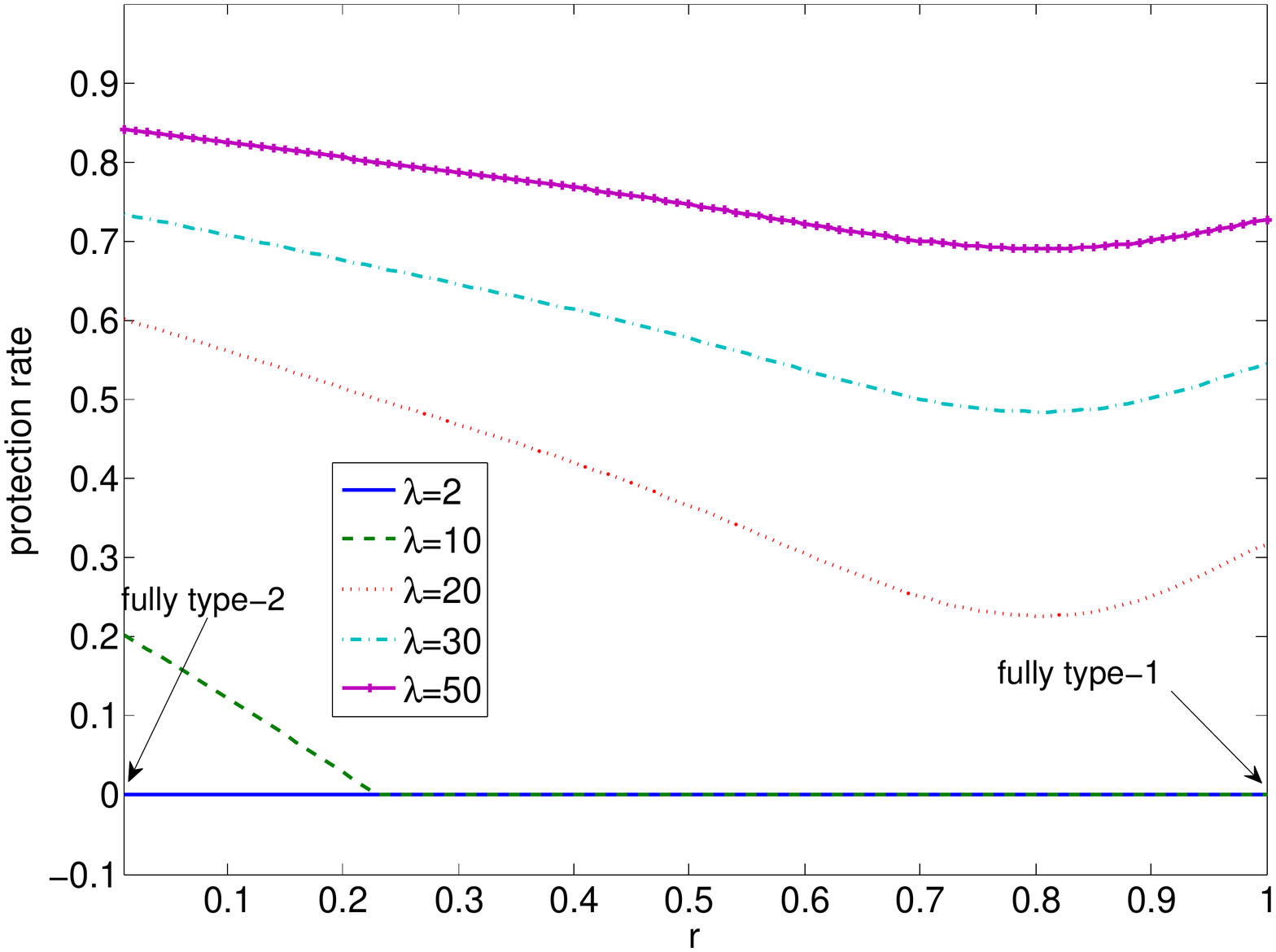}
\caption{Proportion of individuals protected inside the population at equilibrium depending on the average number of nodes $\lambda$ in the graph and two types of individuals with varying the heterogeneity of the population. The parameters are: $\tau_1=0.05$, $\tau_2=0.2$, $C=4$ and $K=5$.}
\label{fig_ex2}
\end{figure}

\begin{figure}[htbp]
\centering
\vspace{-3mm}\includegraphics[width=5in]{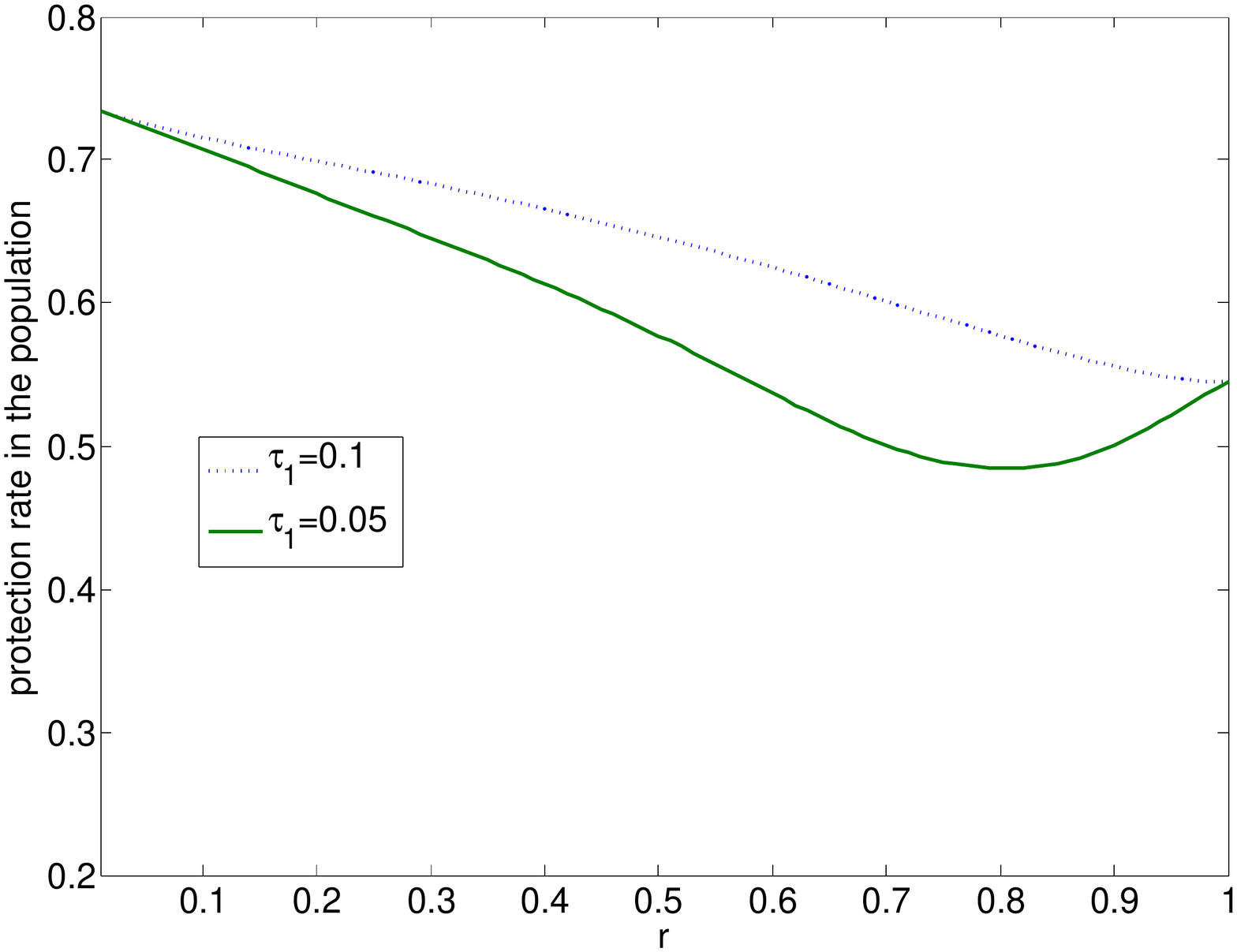}
\caption{Proportion of individuals protected inside the population at equilibrium depending on the proportion $r$ of the two types of individuals varying the heterogeneity of the population. The parameters are: $\lambda=30$, $\tau_2=0.2$, $C=4$ and $K=5$.}
\label{fig_ex3}
\end{figure}


\begin{figure}[htbp]
\centering
\vspace{-3mm}\includegraphics[width=5in]{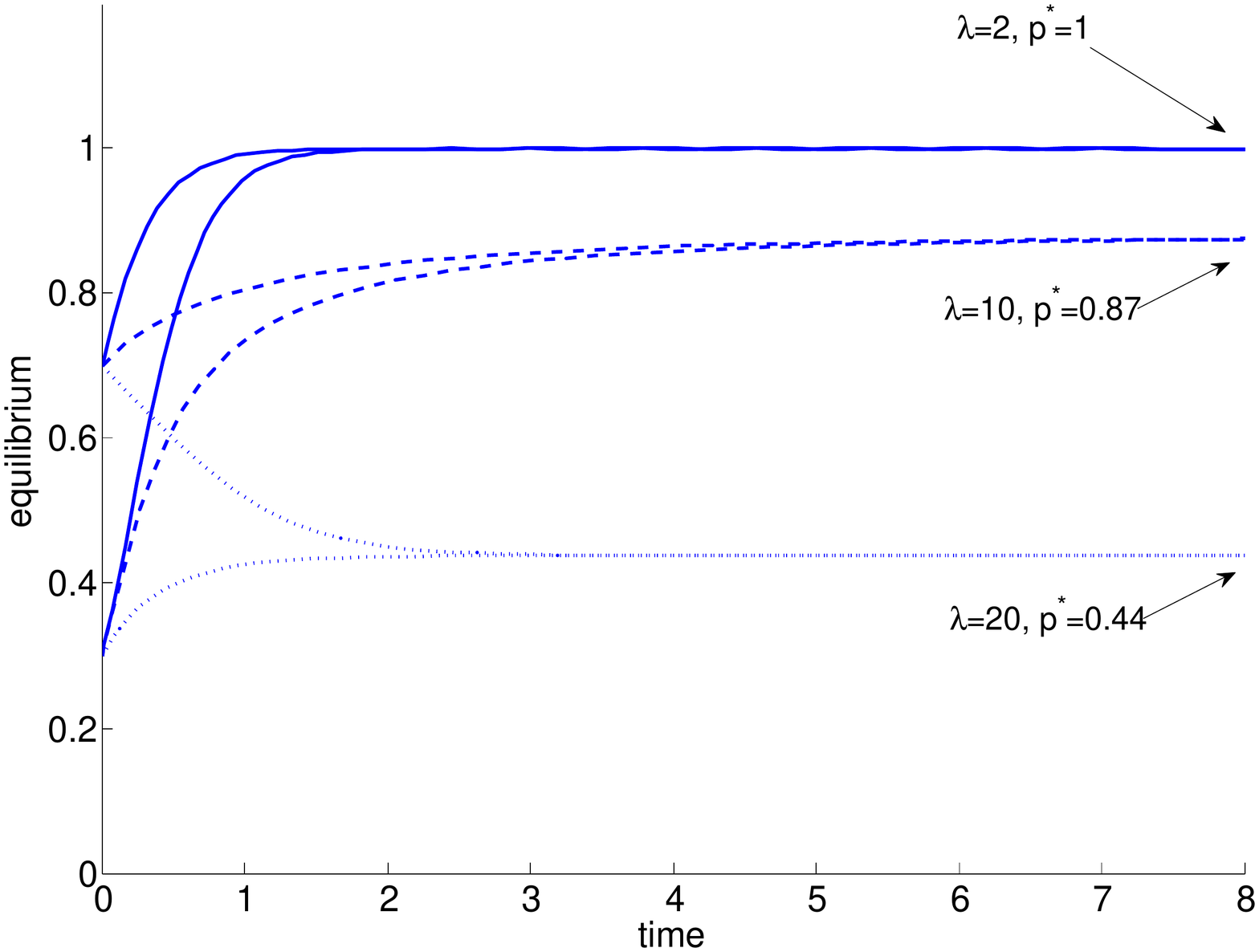}
\caption{Convergence of the replicator dynamics to the Nash equilibrium with 2 types of individuals and the following parameters: $\tau_1=0.05$, $\tau_2=0.2$, $C=4$, $r=0.1$ and $K=5$.. We consider two initial points: $p(0)=0.3$ and $p(0)=0.7$.}
\label{fig_ex4}
\end{figure}

\begin{figure}[htbp]
\centering
\vspace{-3mm}\includegraphics[width=5in]{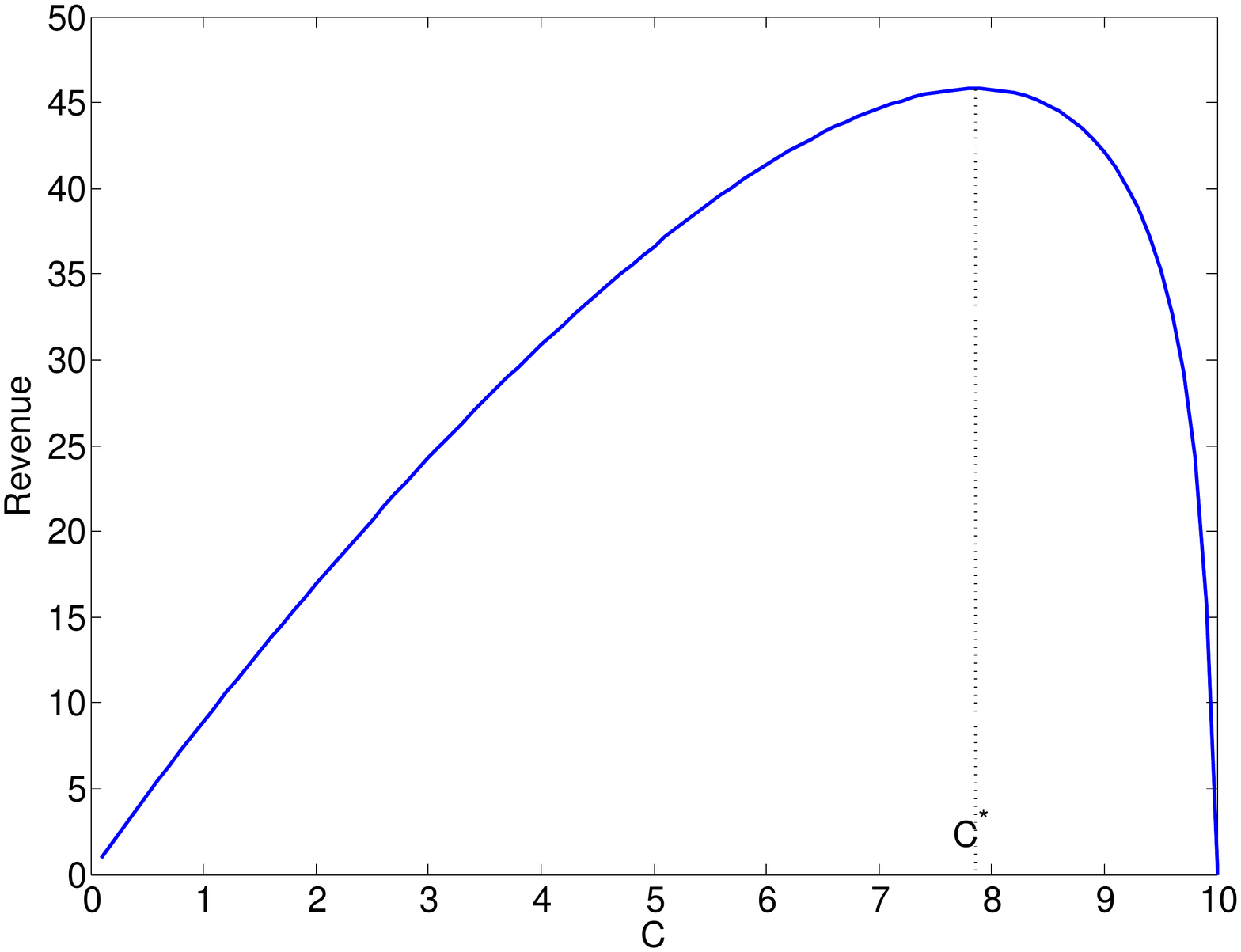}
\caption{Revenue of the provider depending on the price $C$ with the following parameters $\tau_1=0.5$, $\tau_2=0.98$, $K=10$, $r(1)=0.3$ and $\lambda=10$.}
\label{fig_ex5}
\end{figure}


\subsubsection{Equilibrium Paradox}
When $r$ is close to $0$ (resp. $1$), i.e., only type-2 (resp. type-1) individuals form the population, we observe in Fig. \ref{fig_ex2} the impact of both heterogeneity parameters $\lambda$ and the type distribution $r(\cdot)$ on the equilibrium. Specifically, for the reason of convenience, we show the percentage of people protecting themselves (i.e., the protection rate) which corresponds to $1-p^*$. We can observe that 
the average number of interacting individuals, which is equal to $\lambda$, has a positive  impact on the protection rate inside the population. For the same heterogeneity type given by a distribution $r(\cdot)$, the protection rate is strictly increasing with $\lambda$. It is obvious that, when each individual interacts with more individuals in average, it makes individuals more vulnerable to the contagion, and then requires a higher level of protection. 
Comparing $\lambda=2$ (e.g. a population with mostly pairwise interactions as in standard evolutionary game framework) and $\lambda=10$, we can observe that the protection rate is strictly higher only when the proportion of type-2 individuals is higher than around 80\%. This means that under this threshold type proportion, individuals do not feel in a risky environment even if the number of interacting individuals  is large (i.e.,  $\lambda=10$). When $\lambda$ becomes even larger another interesting property arises. From  Fig.  \ref{fig_ex2}, we can observe that, for  parameter $\lambda=20$, the impact of the heterogeneity type is significant. By increasing the heterogeneity from $r=0$ (only type-2 individuals), the protection rate decreases. In fact, increasing the heterogeneity in our scenario means that we reduce the proportion of type-2 individuals meanwhile increasing the proportion of type-1 individuals. 
 
 As we further increase the value $r$,  the phenomenon of \textit{heterogeneity induced confidence principle} arises. A highly heterogeneous population 
 leads to a decreasing the protection rate  as the heterogeneity reaches the percentage value around 20 \% of type-2. Moreover, this threshold seems be independent of the average number of individuals in each local interaction.



We can explain this counter-intuitive outcome by considering the particular case with $\lambda=30$ and with two different effective spreading rate for type-1 individuals. This scenario is depicted on Fig. \ref{fig_ex3}. Here, we can observe that the behavior of the protection rate is as expected, always decreasing when increasing the proportion of type-2 individuals. This is obtained when the effective spreading rate of type-1 individuals is equals to $0.1$. In the other case, when $\tau_1=0.05$ even if more type-2 individuals are in the population, the global protection rate is decreasing as more heterogeneity is introduced. Individuals behave more confidently and protect less themselves. 


\subsubsection{Convergence of the Replicator Dynamics}

In Fig. \ref{fig_ex4}, the replicator dynamics equation is illustrated for two different starting points and also for different values of the average interaction size $\lambda$. This result confirms the convergence of the ODE  to the equilibrium and also that the equilibrium is decreasing with the parameter $\lambda$. In fact, the protection rate is by definition the proportion of individuals that are protected, i.e., $1-p^*$. Then, we can observe on Fig. \ref{fig_ex2} that for $r=0.1$ and $\lambda=10$, the protection rate of the population is $0.13$, which corresponds to the rest point of the replicator dynamics in long dashed line. This result is corroborated for the case where $\lambda=20$ and the protection rate is then $0.56$, we obtain that the replicator dynamics converge to $0.44$. 


\subsubsection{Global optimization and learning algorithm}

Finally, we illustrate the global  control design problem. We first describe in Fig. \ref{fig_ex5}  the global revenue for the controller as a function of his control parameter, and observe the strict concavity property.



We show in Fig. \ref{fig_ex6} the result of our learning mechanism for several control updates with step size $a(n)=\frac{1}{1+n \log (n)}$, $a(n)=\frac{1}{n}$ and $a(n)=\frac{1}{n^2}$, respectively. The first control update satisfies the conditions  in Proposition \ref{condconv}. We can observe that for this control update, our learning algorithm converges to the optimal control, whereas the two other control updates do not. The last control update $a(n)=\frac{1}{n^2}$ gives a too fast update of the control and finally converges to the control value which is far from the optimal one, and therefore, the revenue at this value is far from the optimal revenue.

\begin{figure}[h]
\centering
\vspace{-3mm}\includegraphics[width=5in]{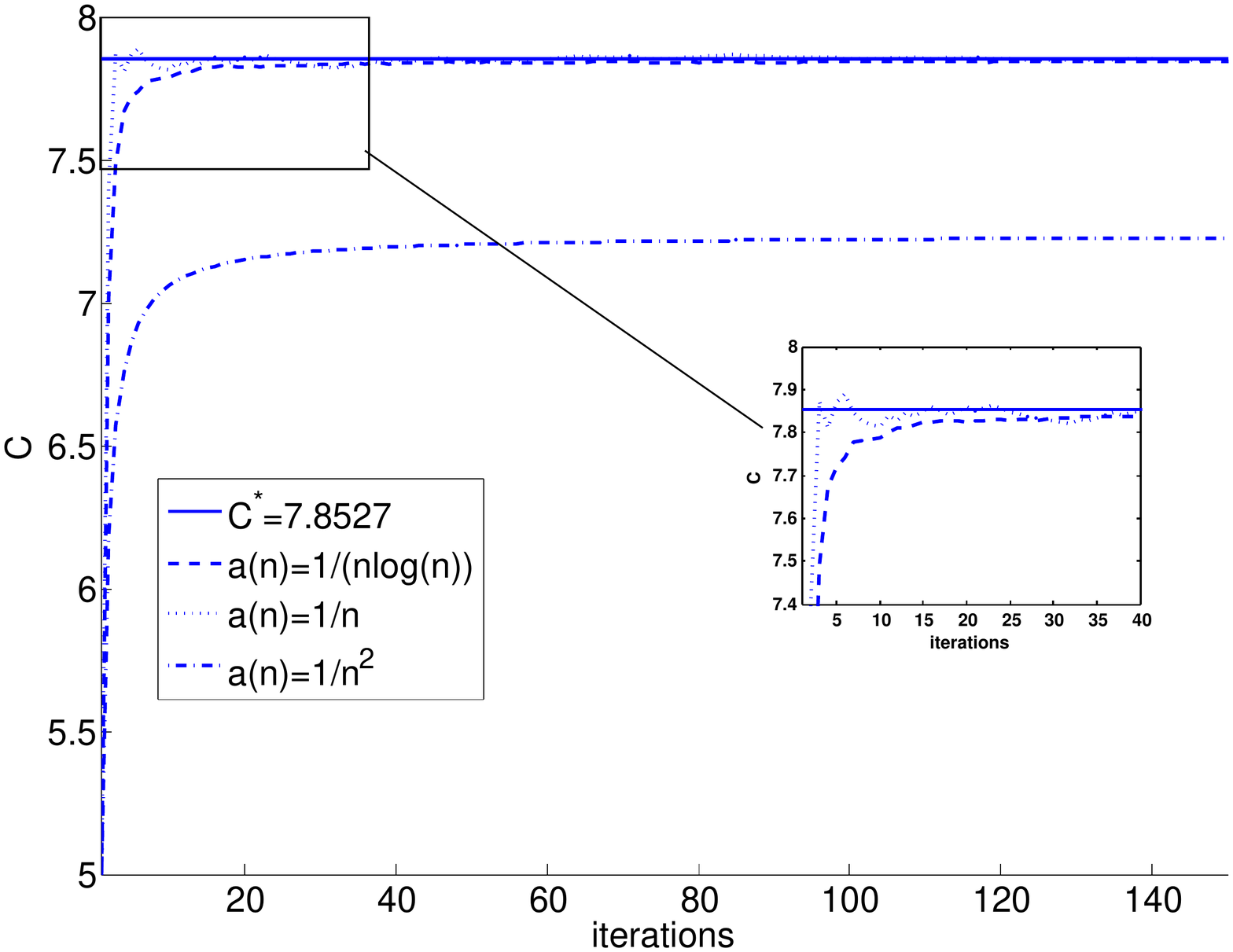}\vspace{-5mm}
\caption{Convergence of our learning algorithm with the SPSA to the optimal price with  $\delta_1=10$, $\delta_2=5.1$, $\beta=5$, $K=10$, $r(1)=0.3$ and $\lambda=10$.}
\label{fig_ex6}
\end{figure}


\section{Conclusion}\label{section:conc}
In this paper, we have first described a framework of  large population game  with heterogeneous types of individuals, in which  local interactions involve a random number of individuals of different types. We have developed the concept of evolutionary stability equilibrium as a solution that characterizes the game behavior. A decentralized virus protection problem has been used to motivate and illustrate this framework. In order to achieve desirable outcome of the game, we have developed methodologies to design a global controller for this dynamic heterogeneous population game. In particular, the interdependency between the global control and the population behaviors has been analyzed using  coupled dynamics between approximate stochastic gradient algorithms with the replicator dynamics. We have shown  the convergence of our learning algorithm, and provided  numerical illustrations to demonstrate the impact of the heterogeneity on the outcome of the system. As a future work, we would lift the assumption of Poisson distribution on the population, and investigate the data-driven reinforcement learning type of algorithms for as a global controller.

\bibliographystyle{ieeetr}

\bibliography{poissongame}

\end{document}